\documentclass[pra,twocolumn,floatfix,superscriptaddress,notitlepage,longbibliography]{revtex4-1}
\usepackage{graphicx} % Required for inserting images

\usepackage[normalem]{ulem}
\usepackage{graphicx, color, graphpap}% Include figure files
\usepackage{enumitem}
\usepackage{amssymb}
\usepackage{amsthm}
\usepackage[dvipsnames]{xcolor}
\definecolor{DarkPurple}{RGB}{119,30,125}
\usepackage[colorlinks=true,urlcolor=DarkPurple,citecolor=DarkPurple,linkcolor=DarkPurple]{hyperref}
\usepackage[T1]{fontenc}
\usepackage{verbatim}
\usepackage{mathtools}
\usepackage{titlesec}
\usepackage{float}
\usepackage{bbm}
\usepackage{tikz}
\usepackage{caption}

\usepackage{hyperref}
\hypersetup{breaklinks=true}

\long\def\ca#1\cb{} %Use for commenting out: \ca...\cb

\newcommand{\ketbra}[2]{| \hspace{1pt} #1 \rangle \langle #2 \hspace{1pt} |}

\newcommand{\norm}[1]{\left\lVert #1 \right\rVert}

\newcommand{\ket}[1]{|#1\rangle}               %ket
\newcommand{\bra}[1]{\langle #1|}              %bra

\newcommand{\FC}{\mathcal{F}}

\newcommand{\OC}{\mathcal{O}}

\renewcommand{\leq}{\leqslant}

\newcommand{\ot}{\otimes}

\renewcommand{\bm}[1]{\boldsymbol{#1}}
\newcommand{\tr}[1]{\mathrm{tr}\left[#1\right]}
\newcommand{\E}[1]{\mathbb{E}\left[#1\right]}
{}%         Body font
{}%         Indent amount (empty = no indent, \parindent 
\newtheorem{theorem}{Theorem}
\newtheorem{lemma}{Lemma}
\begin{document}

\title{Measuring Non-Gaussian Magic in Fermions: Convolution, Entropy, and the Violation of Wick's Theorem and the Matchgate Identity}

\author{Luke Coffman}
\affiliation{Department of Physics, University of Colorado, Boulder, Colorado 80309, USA}
\affiliation{JILA, NIST and University of Colorado, Boulder, Colorado 80309, USA}
\author{Graeme Smith}
\affiliation{Institute for Quantum Computing and Department of Applied Math, University of Waterloo}
\author{Xun Gao}
\affiliation{Department of Physics, University of Colorado, Boulder, Colorado 80309, USA}
\affiliation{JILA, NIST and University of Colorado, Boulder, Colorado 80309, USA}

\begin{abstract}
    Classically hard-to-simulate quantum states, or "magic" states, are prerequisites to quantum advantage, highlighting an apparent separation between classically and quantumly tractable problems. Classically simulable states such as Clifford circuits on stabilizer states, free bosonic states, free fermions, and matchgate circuits are all in some sense Gaussian. While free bosons and fermions arise from quadratic Hamiltonians, recent works have demonstrated that bosonic and qudit systems converge to Gaussians and stabilizers under convolution. In this work, we similarly identify convolution for fermions and find efficient measures of non-Gaussian magic in pure fermionic states. We demonstrate that three natural notions for the Gaussification of a state—(1) the Gaussian state with the same covariance matrix, (2) the fixed point of convolution, and (3) the closest Gaussian in relative entropy—coincide by proving a central limit theorem for fermionic systems. We then utilize the violation of Wick's theorem and the matchgate identity to quantify non-Gaussian magic in addition to a SWAP test.
\end{abstract}
\maketitle
The classical simulability of quantum systems is a question of ubiquitous interest to the quantum information science community. While efficient classical descriptions allow for the detailed analysis of complex systems, challenging-to-simulate quantum systems are necessary to surpass classical information processing capabilities. Consequently, efficient methods of quantifying hard-to-simulate quantum resources offer valuable insight into the apparent gap between classically and quantumly tractable problems.

Since the 1940s, several classes of classically simulable quantum systems have been identified, including unentangled qubit states, Gaussian states in quantum optics~\cite{caves_new_1985}, and free fermions in condensed matter~\cite{lieb_two_1961,barouch_statistical_1971,kaufman_crystal_1949}. A foundational addition to this list were stabilizer circuits composed of Clifford gates by Gottesman and Knill~\cite{gottesman_heisenberg_1998}, demonstrating that entanglement alone is insufficient for classical insimulability. Therefore, the "magic" of quantum computation originates from non-stabilizer elements~\cite{bravyi_universal_2005,knill_fault-tolerant_2004}, making \textit{magic} the resource elevating quantum systems to classically hard.

Clifford circuits map Pauli strings to Pauli strings under conjugation, that is, they normalize the Pauli group~\cite{Nielsen_Chuang_2010}. In linear optics or free bosons, discrete displacements in spin (Pauli operators) are substituted for continuous displacements in phase space (Weyl or displacement operators). Bartlett et al.~\cite{bartlett_efficient_2002} demonstrated that Gaussian circuits normalize the Weyl group of free bosons and are classically simulable mirroring the Gottesman-Knill theorem. Gaussian states are Gaussian in that they are thermal states of Hamiltonians quadratic in the quadratures $x$ and $p$.  Thus, magic in free bosons arises from non-Gaussian elements~\cite{gottesman_encoding_2001}.

Another class of classically-simulable circuits introduced by Valiant are matchgates~\cite{valiant_quantum_2001}. Matchgates turn out to be the Jordan-Wigner transformation of free fermions~\cite{knill_fermionic_2001,terhal_classical_2002,jozsa_matchgates_2008} which are exactly solvable by Bogoliubov transformations~\cite{lieb_two_1961,barouch_statistical_1971,kaufman_crystal_1949}. Free fermions are thermal states of Hamiltonians quadratic in quadratures $a$ and $a^\dagger$, reaffirming the connection between Gaussian properties and classical simulability. Once a non-Gaussian element is introduced, matchgates become universal for quantum computation~\cite{jozsa_matchgates_2008,brod_computational_2013,brod_efficient_2016}, making pure non-Gaussian fermionic states the source of magic~\cite{hebenstreit_all_2019}. Studying magic in fermions is simultaneously vital for quantum computational chemistry~\cite{mcardle_quantum_2020} and an interesting case in between bosons and stabilizers - Gaussian operations on a finite-dimensional space.

Further results on the non-negativity of the Wigner function~\cite{mari_positive_2012} implying Gaussian wave functions over continuous fields for bosons~\cite{hudson_when_1974,soto_when_1983,mandilara_extending_2009} and finite fields for qudits~\cite{gross_hudsons_2006,gross_schurweyl_2021,wootters_wigner-function_1987,wootters_picturing_2003,gross_finite_2005} have made Gaussian properties a key sign of efficient classical simulation. Beyond viewing Gaussians as the exponential of a quadratic form, in classical probability theory, Gaussians are simultaneously the maximum entropy distributions for fixed mean and covariance matrix and invariant under self-convolution. Indeed by the classical Central Limit Theorem (CLT) they are \textit{uniquely} invariant under self-convolution providing an alternative definition of Gaussian. While quantum CLTs have existed since the 1970s~\cite{lenczewski_quantum_1995,Cushen_Hudson_1971,giri_algebraic_1978,hudson_quantum-mechanical_1973,von_waldenfels_algebraic_1978}, only recently has quantum convolution been identified and utilized for quantifying magic in bosons~\cite{becker_convergence_2021,beigi_towards_2023} and qudits~\cite{bu_discrete_2023,bu_quantum_2023,bu_stabilizer_2023}.

In this paper, we define and develop \textit{efficient measures} of non-Gaussian magic in pure fermionic states $\rho$ by studying the difference with its corresponding (possibly mixed) Gaussian state $\rho_G$. Inspired by a geometric interpretation of classical convolution, we define fermionic convolution through a beam splitter interaction in analogy to bosonic and qudit systems. This leads to at least three possible choices for $\rho_G$: (1) the Gaussian state with the same covariance matrix, (2) the Gaussian state obtained as the fixed point of self-convolution, and (3) the closest Gaussian state to $\rho$ in relative entropy. We demonstrate that all three Gaussifications coincide, by analyzing the connected correlation functions or cumulants of $\rho$. Cumulants generalize covariance and capture intrinsic $k$-body interactions not expressible as combinations of low-order correlations. Since Gaussians are completely characterized by their second-order cumulants or covariance matrix and uniquely fixed under convolution, higher-order cumulants $(k\ge 2)$ of $\rho$ must decay to zero under self-convolution. We derive this decay rate for an $m$-mode state $\rho$ after $n$ self-convolutions, thereby obtaining a CLT for fermions. We also study another method of converging to $\rho_G$ by coupling to a Gaussian thermal bath that achieves rapid mixing in $\log m$ time through convolving $\rho$ and $\rho_G$. Finally, since Gaussians have the maximum entropy for a given covariance matrix, the Gaussian with the minimum relative entropy to $\rho$ must have the same covariance matrix as $\rho$, allowing us to refer to $\rho_G$ without ambiguity. From this analysis, we arrive at multiple efficient measures of non-Gaussian magic. The first measure is the entropy of $\rho_G$, the second is the violation of a fermionic matchgate identity derived from convolution, the third is the violation of Wick's theorem~\cite{wick_evaluation_1950}, and the fourth is the probability of failure in a SWAP test.

Before arriving at our measures of non-Gaussian magic, we need an appropriate notion of convolution. Classically, convolving the probability densities of independent random variables $A$ and $B$ corresponds to the probability density for $\frac{1}{\sqrt{2}}(A+B)$ with additional rescaling to preserve the second-order cumulants. Geometrically, this corresponds to rotating the joint probability density for $(A, B)$ clockwise by $\pi/4$ to $\frac{1}{\sqrt{2}}(A+B,A-B)$ and computing the marginal probability density of the first coordinate. In quantum mechanics, density matrices replace probability density while quadratures $a$ and $a^\dagger$ ($b$ and $b^\dagger$) replace the random variable $A$ ($B$). Therefore, the analog of $\tfrac{1}{\sqrt{2}}(a+b)$ for quantum systems should correspond to the convolution of their associated density matrices. Geometrically, we should perform a real rotation of the joint density matrix $\rho_A\ot \rho_B$ clockwise by $\pi/4$ in the $(a,b)$ phase space and compute the marginal density matrix $\rho_A'$. For $m$-degrees of freedom or $m$-modes the rotations should be component-wise just as the classical multivariant convolution is component-wise on $(A_i,B_i)$. Real rotations, i.e. in the x-z plane, are generated by evolution under the second quantization of Pauli $Y$ for fermions or the beam splitter in bosonic quantum optics (derived in Appendix~\ref{app:lemma_fermionic_beam_splitter})
\begin{align}\label{fermionic_beam_splitter}
        U_\eta = \exp\left(-\arccos(\sqrt{\eta})\sum_{j=1}^m(a_jb^\dagger_j+a^\dagger_jb_j)\right),
\end{align}
where $\eta =\cos(t)^2 \in [0,1]$ controls the rotation in phase space or the probability of reflection and transmittance for two light beams. This has the desired defining action of rotating first-moment quantities
\begin{align}
\begin{split}
    U_\eta \bm{a} U_\eta^\dagger &= \sqrt{\eta}\bm{a}+\sqrt{1-\eta}\bm{b},\\
    U_\eta \bm{b} U_\eta^\dagger &= \sqrt{\eta}\bm{b}-\sqrt{1-\eta}\bm{a},
\end{split}
\end{align}
where $\bm{a}=(a_1^\dagger,a_1,\ldots,a_m^\dagger,a_m)$ is the vector of $2m$ creation and annihilation operators. We therefore define the \textit{quantum convolution} of two $m$-mode quantum states $\rho_A$ and $\rho_B$ as
\begin{align}
    \rho_A \boxplus_\eta \rho_B = \mathrm{tr}_B\left[U_\eta (\rho_A \ot \rho_B)U_\eta^\dagger\right],
\end{align}
where $\mathrm{tr}_B[\cdot]$ corresponds to tracing out system $B$. Setting $\eta=1/2$ recovers our classical intuition for convolution: rotate each $(a_i,b_i)$ in phase space clockwise by $\pi/4=\arccos(\sqrt{1/2})$ and compute the first marginal distribution by adding together all possible realizations for system $B$ ($\mathrm{tr}_B[\cdot]$). We can generalize this idea to the convolution of $n$ quantum states by applying consecutive convolutions under different values of $\eta$. Classically, convolving the probability densities of independent random variables $\{X_i\}_{i=1}^n$ corresponds to the probability density for $X^{(n)}=\frac{1}{\sqrt{n}}\sum_{i=1}^n X_i$. By expressing $X^{(n)}$ as the recurrence relation
\begin{align*}
    X^{(k)}=\sqrt{1-\frac{1}{k}}X^{(k-1)}+\frac{1}{\sqrt{k}}X_{k}
\end{align*}
with initial value $X^{(2)}=\frac{1}{\sqrt{2}}(X_1+X_2)$, we recover our geometric intuition for convolution: rotate the joint probability density for $(X^{(k-1)},X_k)$ clockwise by $\arccos(\sqrt{1-1/k})$ and compute the first marginal probability density to find $X^{(k)}$. The natural analog for the convolution of $n$ quantum states is thus,
\begin{align}
    \rho^{\boxplus k}=\rho^{\boxplus (k-1)}\boxplus_{1-\frac{1}{k}}\rho_k \quad \rho^{\boxplus 2}=\rho_1 \boxplus_{1/2} \rho_2.
\end{align}
Note that throughout this work if $\rho^{\boxplus n}$ is written without specifying the states being convolved it is assumed that $\rho=\rho_1=\rho_2=\cdots=\rho_n$ which is the $n$-fold self-convolution of $\rho$.

In our definition of convolution the $\tfrac{1}{\sqrt{2}}$ prefactor was explicitly chosen to preserve second-order cumulants. As such, zero mean Gaussians are fixed under self-convolution due to the covariance matrix completely characterizing the distribution. A reasonable conjecture is that \textit{all} fixed points of convolution are Gaussian. Indeed this is the conclusion of the classical CLT: the $n$-fold self-convolution of a zero mean probability density $p$ approaches a zero mean Gaussian with the same second-order cumulants as $p$. Crucially, the non-Gaussian behavior characterized by higher-order cumulants must decay to zero. Replacing $p$ with $\rho$ suggests a quantum analog of cumulants and convolution provide natural ways of measuring magic due to Gaussian states being classically simulable. As such, let us review key tools from classical probability and search for similar ideas in quantum mechanics.

Classically, the moments and cumulants of a random variable $X$ describe the spread and variability of its probability density $p$. A useful tool for characterizing variations in a function is the Fourier transform. The Fourier transform of the probability density is called the \textit{characteristic function} of $p$ denoted $\chi_p(\bm{\xi})=\mathbb{E}[e^{-i\bm{\xi} X}]$. Each $k$-th derivative of $\chi_p$ at 0 yields the $k$-th moment of $X$, thus Taylor expanding $\chi_p(\bm{\xi})$ generates the moments of $X$. Incidentally, we will use the name \textit{moment generating function} and characteristic function interchangeably. As opposed to moments, cumulants capture intrinsic $k$-body interactions or correlations that cannot be expressed as products of lower-order interactions. Products of low-order interactions arise from independence such as $\E{XY}=\E{X}\E{Y}$ for independent $X$ and $Y$. Utilizing $\log{\E{\cdot}}$ will separate such products into sums,  motivating the study of $\log \chi_p(\bm{\xi})$ to isolate true $k$-th order contributions. Each $k$-th derivative of $\log \chi_p(\bm{\xi})$ yields the corresponding $k$-th order cumulant of $X$ leading to the name \textit{cumulant generating function}.

Rather than taking the expectation of translations on the real line ($e^{-i\bm{\xi}X}$), in quantum mechanics, we should take the expectation of translations in phase space. In Weyl quantization, displacement operators quantize points in a semi-classical phase space. Crucially, the underlying numbers we quantize should mirror the algebraic properties of the original quantum mechanical system. For bosons, there is a continuum of possible commuting configurations for the system making real or complex numbers a natural choice. However, fermions anti-commute ($\{a_i,a_j^\dagger\}=\delta_{i,j}$) and can never occupy the same state due to the Pauli exclusion principle ($\{a_i,a_j\}=0$) resulting in a finite-dimensional Hilbert space. Thus our choice of numbers, known as Grassmann numbers $\xi_i$ ($\xi_i^*$), should anti-commute and square to zero~\cite{hudson_translation-invariant_1980,cahill_density_1999}. Replacing complex numbers for Grassmann numbers in the bosonic displacement operator allows us to define $\chi_\rho(\bm{\xi})$ as the expectation value of the fermionic Weyl or displacement operator $D(\bm{\xi})$,
\begin{align}
    \chi_\rho(\bm{\xi})=\tr{\rho D(\bm{\xi})},D(\bm{\xi})=\exp\left[\sum_i (a_i^\dagger \xi_i-\xi_i^*a_i)\right],
\end{align}
where $\bm{\xi}=(\xi_1^*,\xi_1,\ldots,\xi_m^*,\xi_m)$. The quantity $\tr{\rho D(\bm{\xi})}$ is the operator equivalent of a Fourier transform for $\rho$ between ladder operators and Grassmann numbers. Due to fermionic operators and Grassmann numbers anti-commuting, ambiguities in signs can arise. As such, we use an arbitrary yet fixed ordering denoted $\{a_n^\dagger a_n\}=a_n^\dagger a_n-1/2=-\{a_na_n^\dagger\}$ with trivial definitions $\{a_n^\dagger\}=a_n^\dagger$ and $\{a_n\}=a_n$. Additionally, for non-empty ordered subsets $J\subseteq [2m]=\{1,2,\ldots, 2m\}$, we define $\bm{\xi}^J=\prod_{j\in J}\bm{\xi}_j$, $\bm{a}^J=\{\prod_{j\in J}\bm{a}_j\}$, and the ordered left derivatives of indices in $J$ as $\frac{\partial}{\partial \bm{\xi}^J}$. For instance when $J=\{1,2,4\}$, $\bm{\xi}^J = \xi_1^* \xi_1 \xi_2$, $\bm{a}^J=(a_1^\dagger a_1-1/2) a_2$, and $\frac{\partial}{\partial \bm{\xi}^J}=\frac{\partial}{\partial \xi_2}\frac{\partial}{\partial \xi_1}\frac{\partial}{\partial \xi_1^*}$. For an introduction to this formalism, we refer the reader to the work by Cahill and Glauber~\cite{cahill_density_1999}. Taylor expanding $\chi_\rho(\bm{\xi})$ and $\log[\chi_\rho(\bm{\xi})]$ yield the moments $M_J = \frac{\partial}{\partial \bm{\xi}^J}\chi_\rho(\bm{\xi})\big\rvert_{\bm{\xi}=0} = \tr{\rho \bm{a}^J}$ and cumulants $C_J=\frac{\partial}{\partial \bm{\xi}^J}\log[\chi_\rho(\bm{\xi})]\big\rvert_{\bm{\xi}=0}$ of $\rho$ respectively,
\begin{align}\label{eqn:moment_generating_function}
        \chi_\rho (\bm{\xi}) = 1 + \sum_{1\le|J|}M_J \bm{\xi}^J \quad \log [\chi_\rho (\bm{\xi})]=\sum_{1\le |J|}C_J \bm{\xi}^J.
\end{align}
Through this work, we only consider states with an even number of fermions or convex combinations thereof~\footnote{Since fermions carry half-integer spin, a \textit{physical state} can only change by a global phase when rotated $2\pi$ around any axis. This forces pure states to be a superposition of only an odd number ($-1$ global phase) or only an even number ($+1$ global phase) of fermions. We stick to the convention of only even states. Note: physical density matrices can be convex combinations of both even and odd states, but this is outside the scope of this work.}. Therefore, if the cardinality of $J$ (denoted $|J|$) is odd then $M_J=0$ as odd products of ladder operators net in gaining or losing a fermion. Finally, since Grassmann numbers are at most first order, there are finitely many terms in $\log [\chi_\rho(\bm{\xi})]$ ensuring its convergence.

With suitable definitions of convolution and characteristic functions, we are ready to present the main results of this work. A crucial property of classical convolution is the convolution theorem where the Fourier transform of convolution becomes the product of the Fourier transforms. Here, we recover an analogous quantum convolution theorem
\begin{align}\label{eqn:quantum_convolution_theorem}
        \chi_{\rho_A \boxplus_\eta \rho_B}(\bm{\xi})=\chi_{\rho_A}(\sqrt{\eta}\bm{\xi})\chi_{\rho_B}(\sqrt{1-\eta}\bm{\xi}),
\end{align}
for two $m$-mode states $\rho_A$ and $\rho_B$ (see Appendix~\ref{app:thm_fermionic_convolution_theorem}). From the quantum convolution theorem, we can see quantum convolution preserves second-order cumulants and fixes a state if and only if $C_J=0$ for all $|J|>2$ (i.e. Gaussian). Since $\chi_{\rho^{\boxplus n}}(\bm{\xi})=\chi_\rho(\bm{\xi}/\sqrt{n})^n$, we expect $\rho^{\boxplus n}$ to converge to the Gaussian state $\rho_G$ with the same second-order quantities as $\rho$ which is the statement of a fermionic CLT. Indeed, Taylor expanding the cumulant generating function gives 
\begin{align}
    \log[\chi_{\rho^{\boxplus n}}(\bm{\xi})]=\sum_{1\le|J|}n^{1-|J|/2}C_J\bm{\xi}^J
\end{align}
where $n^{1-|J|/2}C_{J}\rightarrow 0$ for $|J|\ge 4$. When convergence is in terms of trace distance, it is efficient in the number of modes with $\norm{\rho^{\boxplus n}-\rho_G}_1 \le \mathcal{O}\left(m^{12}/n\right)$ (Appendix~\ref{app:proof_fermionic_CLT}). This is reasonable to expect given that the decay rate of the first non-Gaussian term $(|J|=4)$ scales as $1/n$. Physically, the convergence $\rho^{\boxplus n}\rightarrow \rho_G$ can be viewed as coupling $n$-copies of $\rho$ together and allowing them to exchange fermions with probability amplitude controlled by $\eta$. On average, second-order interactions will be unchanged as the coupling ($U_\eta$) is quadratic, making magic the non-trivial information being exchanged between copies. Isolating a single system by averaging out its correlations with the other copies (i.e. taking the partial trace) will destroy the magic correlations spread over the copies by giving the information to the environment. Thus the system will dissipate magic and equilibrate in the $n\rightarrow \infty$ limit to a state that only contains second-order interactions forming a Gaussian. Alternatively, we could imagine coupling a single copy of $\rho$ to a Gaussian thermal bath that comprises many copies of $\rho_G$ through the $t$-fold convolution $\rho \boxplus_\eta \rho_G$. By viewing $(\cdot \boxplus_\eta \rho_G)$ as a quantum channel with $\eta = \exp(-2\alpha)$, the corresponding Lindbladian evolution $e^{t \mathcal{L}_{\eta}}[\rho]$ causes a rapid diffusion of magic into the bath with mixing time $t=\mathcal{O}(\log m)$ as $C_J \rightarrow \exp(-|J|\alpha t)C_J$ for $|J|\ge 4$. Classically, this resembles convolving with the heat kernel to find the solution to the diffusion or heat equation. In either mode of convergence, second-moment quantities are preserved, therefore the equilibrium Gaussian should preserve the maximal amount of information about $\rho$'s second-moment quantities. Indeed, $\rho_G$ is the closest Gaussian state to $\rho$ in relative entropy, where $S(\rho||\rho_G)=S(\rho_G)-S(\rho)$ (Appendix~\ref{app:relative_entropy_Gaussian}). If $\rho$ was pure, then $S(\rho||\rho_G)=S(\rho_G)$ which is an efficient measure of non-Gaussian magic that can be computed by measuring all second-moment quantities of $\rho$. One possible approach to doing this would be to use classical shadows for fermionic systems~\cite{wan_matchgate_2023} and the robustness recently proposed by Bittel et al.~\cite{bittel_optimal_2024} to develop an efficient computation method.

As we have demonstrated, Gaussian states are the unique fixed points of convolution. Thus, Gaussians are invariant under time evolution by the generator of convolution $\Lambda = i\sum_i a_i^\dagger b_i + a_ib_i^\dagger$. Taylor expanding the action of $U(t)=\exp(it\Lambda)$ on a pure state $\ket{\psi}$ implies that $\Lambda \ket{\psi}\ket{\psi}=0$ if and only if $\ket{\psi}$ is Gaussian which is known as the matchgate identity~\cite{cai_theory_2007}. It is worth reiterating here that by only considering the fixed points of convolution we obtain the matchgate identity $\Lambda$. Further, since all Gaussian states vanish after applying $\Lambda$, $\norm{\Lambda \ket{\psi}\ket{\psi}}_2$ is another measure of non-Gaussian magic. However, due to $\Lambda$ being quadratic, this measure contains at most the same information as all second-order cumulants and therefore the same information as $S(\rho_G)$ about non-Gaussian magic. Indeed, it can be shown that
\begin{align}
    \norm{\Lambda \ket{\psi}\ket{\psi}}_2^2=\frac{m}{2}-2\sum_{|J|=2}|M_J|^2=\tfrac{m}{2}-\norm{\Sigma}_2^2,
\end{align}
where $(\Sigma)_{ij}=\tr{\{\bm{a}_i\bm{a}_j\}\rho}$ is the covariance matrix for $\rho$ and $\norm{\cdot}_2$ denotes the Schatten $2$-norm (see Appendix~\ref{app:matchgate}). Under a different ordering, $S(\rho_G)$ is related to the binary entropy of each eigenvalue of $\Sigma$ and thus to first order proportional to $\norm{\Sigma}_2^2$~\cite{surace_scipost_2022,bravyi_lagrangian_2004}. However, rather than estimating all second-order cumulants to compute $S(\rho_G)$, we can estimate the second-order change at $t=0$ of the following (see Appendix~\ref{appendix:lemma_second_order_matchgate})
\begin{align}
    \frac{d^2}{dt^2}\norm{(\exp(it\Lambda)-\mathbb{I})\ket{\psi}\ket{\psi}}_2 \Big|_{t=0}=\norm{\Lambda \ket{\psi}\ket{\psi}}_2^2.
\end{align}

Due to the previous two measures of non-Gaussian magic depending only on second-order cumulants, it is natural to want to quantify higher-order $C_J$. One way is to directly measure $C_J$ then compare to Wick's theorem. Since second-order cumulants completely characterize a Gaussian distribution, higher-order moments provide no additional information. This is captured by Wick's theorem~\cite{wick_evaluation_1950} in quantum field theory (or Isserlis' theorem~\cite{isserlis_formula_1918} in statistics) which for Gaussian $\rho$ decomposes higher-order moments into a sum over the product of second-order moments given by $M_J =\mathrm{Pf}[\Sigma\vert_{J}]$
where $\mathrm{Pf}[\cdot]$ is the Pfaffian of the submatrix $\Sigma \vert_J$ given by indices $J$ of the covariance matrix $\Sigma$~\cite{bravyi_lagrangian_2004}. In quantum field theory, Wick's theorem corresponds to the two-point correlation function for particles in $J$ represented by Feynman diagrams. For instance, $J=[4]$ gives
\begin{align}
    M_{\{1,2,3,4\}} = M_{12}M_{34}-M_{13}M_{24}+M_{14}M_{23}
\end{align}
where we note the sign difference due to a crossing of terms (i.e. $1234\mapsto\xi_1^*\xi_1\xi_2^*\xi_2=-\xi_1^*\xi_2\xi_1\xi_2^*$, see Appendix~\ref{app:Wicks_Cumulants}). To compare this with the true cumulants of $\rho$, we Taylor expand $C_J = \frac{\partial }{\partial \bm{\xi}^J}\log [\chi_\rho(\bm{\xi})]\vert_{\bm{\xi}=0}$ to obtain the following,
\begin{align}
    C_J=\sum_{l=1}^{|J|} \frac{(-1)^{l-1}}{l}\sum_{\lambda \vdash_l J}\mathrm{Sgn}(\pi_\lambda)\prod_{i=1}^l M_{\lambda_i}
\end{align}
where $\lambda \vdash_l J$ denotes the ordered partitions $\lambda$ of $J$ with $l$ pieces and $\mathrm{Sgn}(\pi_\lambda)$ is the sign of the permutation reordering $\lambda$ into $J$. This formula resembles that of classical probability theory with the addition of the $\mathrm{Sgn}(\pi_\lambda)$ to account for anti-commutation relations~\cite{peccati_wiener_2011}. We enumerate the first few cumulants below:
\begin{align}
\begin{split}
    C_{[1]} &= M_{[1]} = 0,\\
    C_{[2]} &= M_{12}-M_{1}M_{2} = M_{[2]},\\
    C_{[3]} &= 0,\\
    C_{[4]} &= M_{[4]}-M_{12}M_{34}+M_{13}M_{24}-M_{14}M_{23}.
\end{split}
\end{align}
Notice for $C_{[4]}$ this quantity is zero if the underlying distribution is Gaussian due to Wick's theorem. In quantum field theory, the cumulants are called the connected correlation function and isolate intrinsic $k$-body interactions~\cite{helias_statistical_2020}. For high-order cumulants, the number of terms that must be measured grows rapidly, however, low-order terms or weak interactions can be measured efficiently.

So far, relative entropy and the matchgate identity only depend on second-order cumulants, in contrast, Wick's theorem captures magic in higher-order cumulants but is challenging to compute. An ideal middle ground would capture non-Gaussian magic in a single quantity that can be directly estimated. Since convolution is a quantum channel and $\rho_G$ is the fixed point, by the data processing inequality we have that
\begin{align}
    S(\rho \boxplus \rho || \rho_G)&\le S(\rho ||\rho_G),\\
    S(\rho \boxplus_\eta \rho_G || \rho_G)&\le S(\rho||\rho_G),
\end{align}
where equality holds if and only if $\rho=\rho_G$~\cite{Nielsen_Chuang_2010}. Due to $\rho$ being pure, $S(\rho ||\rho_G)=S(\rho_G)$, thus convolution produces entropy in $\rho\boxplus \rho$ and $\rho \boxplus_\eta \rho_G$ if and only if $\rho$ is non-Gaussian. The additional fact that $S(\rho^{\boxplus n}||\rho_G)\ge S(\rho^{\boxplus}||\rho_G)$ (similarly for $\boxplus_\eta$) also due to the data processing inequality is the statement of the second law of thermodynamics from an information-theoretic perspective. Namely, the system becomes more mixed with each additional convolution if $\rho$ isn't Gaussian. As in the classical case of the state of a Markov chain, each additional convolution can be viewed as a random walk of the fermions. Thus, to measure non-Gaussian magic, we can quantify the overlap of $\rho \boxplus \rho$ or $\rho \boxplus_\eta \rho_G$ with $\rho$ by performing a $\mathrm{SWAP}$ test. The $\mathrm{SWAP}$ operator is given by
\begin{align}
    \mathrm{SWAP}&=\prod_{i}\left(\tfrac{1}{2}\mathbb{I}+a_i^\dagger b_i + b_i^\dagger a_i + 2\{a_i^\dagger a_i\}\{b_i^\dagger b_i\}\right).
\end{align}
Evaluating this operator on a pure state $\rho$ and expanding we see it has the following functional form
\begin{align}
    \tr{\mathrm{SWAP} \rho \ot \rho}=\sum_{J\subseteq [2m]} \delta(J) |M_J|^2 = 1,
\end{align}
where equality to one follows from $\rho$ being pure and $\delta(J)$ tracks factors and signs from anti-commuting operators (see Appendix~\ref{app:sec_SWAP_test}). The formula for $\delta(J)$ is unimportant as it depends only on $J$ and is thus solely for bookkeeping. Crucially, this expression implies that high-order moments can be related to second-order moments such that useful information about higher-order cumulants can be captured. Utilizing Eq.~\eqref{eqn:moment_generating_function} and Eq.~\eqref{eqn:quantum_convolution_theorem}, yields
\begin{align}
    \tr{\mathrm{SWAP}\rho \ot \rho^{\boxplus 2}}= \sum_{J\subseteq [2m]}\delta(J)2^{1-|J|/2}|M_J|^2
\end{align}
for the $\mathrm{SWAP}$ test on $\rho$ and $\rho^{\boxplus 2}$. Therefore $|1-\tr{\mathrm{SWAP}\rho \ot \rho^{\boxplus 2}}|$ is a measure of non-Gaussian magic that captures information about high-order cumulants while being efficient to estimate. 

In this work, we found an appropriate notion of fermionic convolution inspired by a geometric interpretation of classical convolution. We then demonstrated that three instances of the Gaussification of $\rho$ coincide most notably through a fermionic central limit theorem. While the convergence in trace distance achieved the classical dependence of $\tfrac{1}{n}$, an open question is to find the optimal dependence on $m$ in both trace distance and relative entropy. With this work, it is clear that the dependence is at most polynomial in $m$. Motivated by three instances we proposed several efficient measures of non-Gaussian magic. This included the entropy of $\rho_G$, the violation of the matchgate identity $\Lambda$, the violation of Wick's theorem, and the SWAP test. These different measures were then compared in terms of their ease of estimation and the information they capture about the cumulants of $\rho$. An open question is the behavior of these measures in many-body systems, particularly their ability to diagnose new phases of matter~\cite{bejan_dynamical_2024}. On the other hand, with magic quantifying the potential of states for quantum computation further work can be made on non-Gaussian rank which is of key interest to quantum chemistry~\cite{mcardle_quantum_2020}.

While the presented measures provide insight into the landscape of possible strategies, the development of measures that extract more information about high-order cumulants is of key interest. In particular, estimating quantities of the form $\sum_J |J| |M_J|^2$, known as the influence in Boolean functional analysis~\cite{odonnell_analysis_2014}, may provide useful information about magic. However, based on the discussions of this work, measures of the form $\sum_J |J||C_J|^2$ are more powerful as they isolate intrinsic $k$-body interactions within the system. It may be possible for convolution to provide a polynomial extract of this information to all orders for pure states. The more general problem of estimating measures of the form $\sum_J f(|J|)|C_J|^2$ would encapsulate these ideas and shed light on the utility of fermions for the simulation of complex problems. 

Finally, the technical tools of convolution, characteristic functions, cumulant generating functions, and Grassmann algebras utilized in this work may be of independent interest. These tools find wide utility in physics through phase space methods, however, their application to quantum information theory appears to be underutilized. For instance, it may be possible to extend this work to parafermions and generalize non-Gaussian behavior and matchgates to qudit systems with the help of a generalized Jordan-Wigner transformation~\cite{yao_parafermionization_2021,batista_generalized_2001}. While the theory and formulation of parafermions is largely open and has prove be challenging this direction may prove fruitful for both endeavors. A more concrete application for quantum information theory is the simulation of noisy circuits in contrast to previous methods employing characteristic functions for stabilizer states known as the Pauli path. 

\emph{Acknowledgements---} X.G. acknowledges support from NSF PFC grant No.\ PHYS 2317149.

\emph{Note---} During the preparation of this manuscript the authors became aware of another work studying fermionic convolution, the central limit theorem, and the application to measures of non-Gaussian magic~\cite{lyu_fermionic_2024}. In Appendix~\ref{app:lyu_connection} we outline the connection between the two approaches and definitions of characteristic functions.

\bibliography{main.bib}

\setcounter{section}{0}
\setcounter{proposition}{3}
\setcounter{theorem}{4}
\setcounter{lemma}{1}
\setcounter{corollary}{1}
\setcounter{figure}{0}
\renewcommand{\figurename}{Sup. Fig.}

\onecolumngrid
\appendix
\tableofcontents
\section*{Appendix}
\maketitle

\section{Review of Fermionic Linear Optics and Grassmann Numbers}
Here we provide a brief overview of the essential tools from free fermions or fermionic linear optics introduced by Cahill and Glauber~\cite{cahill_density_1999}. Let $a_j^\dagger$ and $a_j$ denote creation and annihilation operators on the $j$-th mode of a fermionic system comprised of $m$-modes, then we have the following defining relations 
\begin{align}
    \{a_j,a_k^\dagger\}=\delta_{jk}, \quad \{a_j,a_k\}=0, \quad \{a_j^\dagger, a_k^\dagger \}=0,\quad a_j\ket{0}=0, \quad \bm{a}=(a_1^\dagger,a_1,\ldots,a_m^\dagger,a_m)
\end{align}
where $\ket{0}$ is the vacuum state. Due to the anti-commutativity of fermionic operators, we define a specific ordering of raising and lowering operators parameterized by $s$ in addition to the $s$ order product of ladder operators for nonempty $J\subseteq [2m]=\{1,2,\ldots,2m\}$,
\begin{align}
    \{a_j^\dagger a_j\}_s=a_j^\dagger a_j + \tfrac{1}{2}(s-1), \quad \{a_j^\dagger\}_s = a_j^\dagger, \quad \{a_j\}_s=a_j, \quad \bm{a}^J_s = \Big\{\prod_{j\in J}\bm{a}_j\Big\}_s.
\end{align}
Here $\bm{a}^J_s$ is the product of ladder operators with $s$ ordering, for instance, $J=\{1,3,4\}$ gives $\bm{a}^J_s = \{a_1\}_s\{a_2^\dagger a_2\}_s=a_1 (a^\dagger_2 a_2 +\tfrac{1}{2}(s-1))$. Throughout this work, we will often focus on the $s=0$ or the symmetric ordering in which case we will drop the $s$ subscript. Similarly, we define the set of anti-commuting Grassmann numbers $\xi_j$ and their complex conjugates $\xi^*_j$ by the relations and additional definitions,
\begin{align}
    \{\xi_j,\xi_k^*\}=0,\quad \{\xi_j,\xi_k\}=0,\quad \{\xi_j^*,\xi_k^*\}=0, \quad \{\xi_j, a_k\}=0, \quad \bm{\xi}=(\xi_1^*,\xi_1,\ldots, \xi_m^*,\xi_m), \quad \bm{\xi}^J=\prod_{j\in J}\bm{\xi}_j.
\end{align}
These numbers naturally capture the characteristics of fermions and are quantized by fermionic operators. Since the square of a Grassmann number is zero, the most general function $f(\xi)$ of a single anticommuting variable is $f(\xi)=u+\xi t$. The left derivative of $f$ is $\tfrac{\partial f}{\partial \xi}=t$ where if $t$ was anti-commuting and $f=u+t\xi$ then $\tfrac{\partial f}{\partial \xi}=-t$ due to first anti-commuting $\xi$ and $t$ then taking the left derivative. For a given non-empty subset $J\subset [2m]$ we define the derivative with respect to the $J$ Grassmann numbers as $\tfrac{\partial}{\partial \bm{\xi}^J}$. For instance $J=\{1,2,3\}$ gives $\tfrac{\partial}{\partial \bm{\xi}^J}=\tfrac{\partial}{\partial \xi_2}\tfrac{\partial}{\partial \xi_1}\tfrac{\partial}{\partial \xi_1^*}$. Integration over Grassmann numbers is defined as follows
\begin{align}
    \int d \alpha_j = \int d\alpha_j^* = 0 \quad \int d\alpha_j \alpha_k = \delta_{jk} \quad \int d \alpha_j^* \alpha_k^* = \delta_{jk}.
\end{align}
With this definition, left differentiation and integration with respect to Grassmann numbers are equivalent. We will often be integrating all Grassmann numbers and their complex conjugates and will utilize the following compact notation,
\begin{align}
    \int d^2 \alpha_j = \int d\alpha_j^* d\alpha_j \quad d\alpha_j d\alpha_j^* = -d\alpha_j^* d\alpha_j \quad \int d^2 \bm{\alpha} = \int \prod_j d^2 \alpha_j.
\end{align}
In analogy to the free bosons or Gaussian quantum optics, we define the $s$-ordered unitary displacement or Weyl operator $D(\bm{\xi},s)=\{D(\bm{\xi})\}_s$ and coherent state by,
\begin{align}
    D(\bm{\xi},s)=\exp\left[\sum_j (a_j^\dagger \xi_j-\xi_j^*a_j+\tfrac{s}{2}\xi_j^*\xi_j)\right], \quad D^\dagger(\bm{\xi},s)\bm{a}D(\bm{\xi},s)=\bm{a}+\bm{\xi}, \quad \ket{\bm{\xi}}_s=D(\bm{\xi},s)\ket{0}.
\end{align}
Crucially the displacement operators form a basis for the space of operators on our Hilbert space, however, they are not orthogonal in the Hilbert-Schmidt inner product. The orthogonal operators are defined by $\delta(\bm{\gamma}-\bm{\xi})=\tr{D(\bm{\gamma},s)E(\bm{\xi},s)^\dagger}$ where
\begin{align}
    E(\bm{\xi},s)=\int d^2 \bm{\alpha} \exp\left[\sum_j (\xi_j\alpha_j^*-\alpha_j\xi_j^*+\tfrac{s+1}{2}\xi_j^*\xi_j)\right]\ketbra{\bm{\alpha}}{-\bm{\alpha}},
\end{align}
which is the Fourier transform of $\ketbra{\bm{\alpha}}{-\bm{\alpha}}$ with $s$-ordering. Since the displacement operators are complete so are the $E(\bm{\xi},s)$. Directly computing $E(\xi,s)$ gives
\begin{align}
    E(\xi,s)=2(\tfrac{1}{2}-a^\dagger a)+(s+1)\xi^*\xi(\tfrac{1}{2}-a^\dagger a)-\xi^*\xi aa^\dagger +\xi a^\dagger -\xi^* a,
\end{align}
which allows us to express the ladder operators as follows,
\begin{align}\label{app:eqn_ladder_E_operator}
    a = \int d^2 \xi (-\xi)E(\xi,s) \quad a^\dagger  = \int d^2 \xi (-\xi^*)E(\xi,s) \quad \{a^\dagger a\}_0=\int d^2 \xi (-\tfrac{1}{2}\xi \xi^*) E(\xi,s).
\end{align}
As in the bosonic case, we can express any operator such as a density matrix $\rho$ in terms of its characteristic function as,
\begin{align}
    \rho = \int d^2 \bm{\xi} \chi_\rho (\bm{\xi},s)E(-\bm{\xi},-s), \quad \chi_\rho(\bm{\xi},s)=\tr{\rho D(\bm{\xi},s)},
\end{align}
which is the quantum analog of the Fourier inversion formula. Taylor expanding the characteristic function or moment-generating function and $\log[\chi_\rho(\xi,s)]$ known as the cumulant generating function give
\begin{align}\label{app:eqn_moment_generating_function}
        \chi_\rho (\bm{\xi},s) = 1 + \sum_{1\le|J|}M_J \bm{\xi}^J \quad \log [\chi_\rho (\bm{\xi},s)]=\sum_{1\le |J|}C_J \bm{\xi}^J.
\end{align}
where $M_J = \frac{\partial}{\partial \bm{\xi}^J}\chi_\rho(\bm{\xi},s)\big\rvert_{\bm{\xi}=0} = \tr{\rho \bm{a}^J_s}$ are the moments of $\rho$ and the $C_J$ are the cumulants or connected correlation function of $\rho$. We will assume that $\rho$ is pure and even, that is when the cardinality of $J$ (denoted $|J|$) is odd then $M_J=C_J=0$.

\section{Correspondence of Three Notions of Closest Gaussian}
In classical probability, there are many equivalent notions of what it means to be Gaussian. Here we focus on the following three. First, the probability density is completely described by first and second-order cumulants (mean and covariance matrix). Second, the probability density is invariant under convolution with itself (by the uniqueness of the central limit theorem or CLT). Third, the probability density achieves the maximum differential entropy for all distributions with the same first and second-order cumulants. Thus for a general probability density $p$, we can define the closest Gaussian distribution to it $p_G$ as either: 1. that given by the first and second-order cumulants of $p$, 2. the one obtained by convolving $p$ with itself many times, or 3. the closest Gaussian to $p$ in differential entropy without ambiguity. In quantum mechanics, we replace $p$ with $\rho$ and $p_G$ with $\rho_G$; however, a priori, we do not know if these three notions of closest Gaussian again coincide. In this section, we define a fermionic quantum convolution and demonstrate that these three notions of closest Gaussian do coincide. To do so, we prove that $(1)\Longleftrightarrow(2)$ through a fermionic CLT and $(1)\Longleftrightarrow(3)$ through direct calculation.
\subsection{Fermionic Convolution and Quantum Convolution Theorem}
Classically, convolving the probability densities of independent random variables $A$ and $B$ corresponds to rotating the joint density for $(A,B)$ by $\pi/4$ clockwise to $\tfrac{1}{\sqrt{2}}(A+B,A-B)$ then compute the marginal density matrix for the first variable. In quantum mechanics, this corresponds to a rotation with real coefficients, i.e. in the XZ plane, of the joint density matrix $\rho_A\ot \rho_B$ clockwise by $\pi/4$ and computing the marginal density matrix for the first system. This unitary rotation corresponds to a beam splitter operation in fermionic systems which we now derive.
\begin{lemma}[Fermionic Beam Splitter Eq.~\eqref{fermionic_beam_splitter}]\label{app:lemma_fermionic_beam_splitter}
    The unitary operator $U_\eta$ for $\eta\in [0,1]$ that transforms
    \begin{align}
        U_\eta \bm{a}U_\eta^\dagger = \sqrt{\eta}\bm{a}+\sqrt{1-\eta}\bm{b},\\
        U_\eta \bm{b} U_\eta^\dagger = \sqrt{\eta}\bm{b}-\sqrt{1-\eta}\bm{a}
    \end{align}
    is given by
    \begin{align}
        U_\eta = \exp \left(-\arccos(\sqrt{\eta})\sum_{j=1}^m (a_jb_j^\dagger + a_j^\dagger b_j)\right).
    \end{align}
\end{lemma}
\begin{proof}
For simplicity, we focus on the case where $\rho_A$ and $\rho_B$ are single-mode systems as the $m$-mode case will be the $m$-fold tensor product of the single mode case. Since we are requiring the rotation between the modes to be real the evolution after time $t$ is as follows
\begin{align}
\begin{split}
    a(t)&=\cos(t)a+\sin(t)b=e^{it \Lambda}ae^{-it\Lambda}\\
    b(t)&=-\sin(t)a+\cos(t)b=e^{it \Lambda}be^{-it\Lambda}
\end{split}
\end{align}
where $\Lambda$ is the Hamiltonian generating this rotation. Solving for the Heisenberg evolution of the operators we have that,
\begin{align}
\begin{split}
    \frac{d}{dt}a(t)&=i[\Lambda,a]=b(t) \implies i[\Lambda,a]=b\\
    \frac{d}{dt}b(t)&=i[\Lambda,b]=-a(t) \implies i[\Lambda,b]=-a.
\end{split}
\end{align}
Solving for $\Lambda$ we find that $\Lambda = i(ab^\dagger + a^\dagger b)$. Therefore our unitary transformation is
\begin{align}
    U=\exp\left[it\Lambda \right] = \exp\left[-t(ab^\dagger + a^\dagger b)\right],
\end{align}
where setting $\cos(t)=\sqrt{\eta}$ gives us our desired result.
\end{proof}

\begin{theorem}[Convolution of Fermionic Characteristic Functions~\eqref{eqn:quantum_convolution_theorem}]\label{app:thm_fermionic_convolution_theorem}
    Given two quantum states $\rho$ and $\sigma$, the characteristic function of their convolution is,
    \begin{align}
        \chi_{\rho \boxplus_\eta \sigma}(\bm{\xi},s)=\chi_{\rho}(\sqrt{\eta}\bm{\xi},s)\chi_\sigma(\sqrt{1-\eta}\bm{\xi},s).
    \end{align}
\end{theorem}
\begin{proof}
    We begin with the representation of $\rho$ and $\sigma$ in terms of their characteristic functions:
    \begin{align}
        \rho = \int d^2\bm{\xi}\chi_\rho(\bm{\xi},s)E(-\bm{\xi},-s), \quad \sigma = \int d^2 \bm{\gamma} \chi_\sigma(\bm{\gamma},l)E(-\bm{\gamma},-s).
    \end{align}
    To compute $U_\eta(\rho \otimes \sigma)U_\eta^\dagger$, we simplify our focus to $U_\eta[E(-\bm{\xi},-s)\otimes E(-\bm{\gamma},-s)]U_\eta^\dagger$ due to $U_\eta$ commuting with Grassmann numbers. This gives us the following,
    \begin{align}
    \begin{split}
        &U_\eta[E(-\bm{\xi},-s)\otimes E(-\bm{\gamma},-s)]U_\eta^\dagger=\\ 
        &\int d^2\bm{\alpha}d^2\bm{\beta}\exp\left(\sum_j -\xi_j\alpha_j^*+\alpha_j\xi_j^*-\gamma_j\beta_j^*+\beta_j\gamma_j^* + \tfrac{-s+1}{2}(\xi_j^*\xi_j + \gamma_j^*\gamma_j)\right)U_\eta\ketbra{\bm{\alpha}\bm{\beta}}{-\bm{\alpha}\bm{\beta}}U_\eta^\dagger.
    \end{split}
    \end{align}
    Now we compute $U_\eta\ketbra{\bm{\alpha}\bm{\beta}}{\bm{\alpha}\bm{\beta}}U_\eta^\dagger$ by using the fact that $U_\eta$ is unitary and that $U_\eta \ket{0}=\ket{0}$ yielding us,
    \begin{align}
        U_\eta\ketbra{\bm{\alpha}\bm{\beta}}{-\bm{\alpha}\bm{\beta}}U_\eta^\dagger &= U_\eta [D(\bm{\alpha})\otimes D(\bm{\beta})]\underbrace{U_\eta^\dagger U_\eta}_{\mathbb{I}}\ketbra{\bm{0}\bm{0}}{\bm{0}\bm{0}} \underbrace{U_\eta^\dagger U_\eta}_{\mathbb{I}} [D^\dagger(\bm{-\alpha})\otimes D^\dagger(\bm{-\beta})]U_\eta^\dagger,\\
        &=U_\eta [D(\bm{\alpha})\otimes D(\bm{\beta})]U_\eta^\dagger \ketbra{\bm{0}\bm{0}}{\bm{0}\bm{0}} U_\eta [D^\dagger(\bm{-\alpha})\otimes D^\dagger(\bm{-\beta})]U_\eta^\dagger.
    \end{align}
    We again shift our focus to computing $U_\eta [D(\bm{\alpha})\otimes D(\bm{\beta})]U_\eta^\dagger$, here we can employ Lemma~\ref{app:lemma_fermionic_beam_splitter} and the fact that $Ue^XU^\dagger=e^{UXU^\dagger}$ for all $X$ and unitary $U$ to see the following,
    \begin{align}
    \begin{split}
        U_\eta\left[a_{j,1}^\dagger\alpha_j-\alpha_j^*a_{j,1}+a_{j,2}^\dagger \beta_j -\beta_j^*a_{j,2}\right]U_\eta^\dagger&=a_{j,1}^\dagger(\sqrt{\eta}\alpha_j+\sqrt{1-\eta}\beta_j)-(\sqrt{\eta}\alpha_j^*+\sqrt{1-\eta}\beta_j^*)a_{j,1}\\
        &+a_{j,2}^\dagger(-\sqrt{1-\eta}\alpha_j+\sqrt{\eta}\beta_j)-(-\sqrt{1-\eta}\alpha_j^*+\sqrt{\eta}\beta_j^*)a_{j,2}.
    \end{split}
    \end{align}
    Thus, $U_\eta [D(\bm{\alpha})\otimes D(\bm{\beta})]U_\eta^\dagger = D(\sqrt{\eta}\bm{\alpha}+\sqrt{1-\eta}\bm{\beta})\otimes D(-\sqrt{1-\eta}\bm{\alpha}+\sqrt{\eta}\bm{\beta})$. Define $\bm{v}=\sqrt{\eta}\bm{\alpha}+\sqrt{1-\eta}\bm{\beta}$ and $\bm{w}=-\sqrt{1-\eta}\bm{\alpha}+\sqrt{\eta}\bm{\beta}$. Further, let $B$ be the orthogonal transformation, i.e. $B^T B=\mathbb{I}_2$, such that $B\alpha_j=v_j$ and $B\beta_j=w_j$. Then to simplify the exponential term of our original integral, observe that the terms are the result of a dot product which is invariant under orthogonal transforms that is,
    \begin{align}
        \sum_j \left(-\xi_j\alpha_j^* + \alpha_j\xi_j^*-\gamma_j \beta_j^*+\beta_j\gamma_j^*+\tfrac{-s+1}{2}(\xi_j^*\xi_j + \gamma_j^*\gamma_j)\right)=\sum_j(-\xi_j' v_j^* + v_j{\xi_j'}^*-\gamma_j'w_j^* + w_j {\gamma_j'}^*+\tfrac{-s+1}{2}({\xi_j'}^*\xi_j'+{\gamma_j'}^*\gamma_j'))
    \end{align}
    where $-\xi'_j=-B\xi_j=-(\sqrt{\eta}\xi_j+\sqrt{1-\eta}\gamma_j)$ and $-\gamma'_j=-B\gamma_j=-(\sqrt{1-\eta}\xi_j+\sqrt{\eta}\gamma_j)$. Additionally, since $\det(B)=1$ and Grassmann integration or Berezin integration multiples by the determinant of the Jacobian, this transformation preserves normalization. Thus the beam splitter operator has the total net effect,
    \begin{align}
        U_\eta [E(-\bm{\xi},-s)\ot E(-\bm{\gamma},-s)]U_\eta^\dagger = E[-(\sqrt{\eta}\bm{\xi}+\sqrt{1-\eta}\bm{\gamma}),-s]\ot E[-(\sqrt{1-\eta}\bm{\xi}+\sqrt{\eta}\bm{\gamma}),-s]
    \end{align}
    leading us to the overall action of the beam splitter to be,
    \begin{align}
        U_\eta(\rho \ot \sigma)U_\eta^\dagger = \int d^2\bm{\xi}d^2\bm{\gamma} \chi_\rho(\bm{\xi},s)\chi_\sigma(\bm{\gamma},s) E(-\bm{\xi'},-s)\ot E(-\bm{\gamma'},-s).
    \end{align}
    Thus, by taking the partial trace of the second system we are left with $\rho \boxplus_\eta \sigma$,
    \begin{align}
        \rho \boxplus_\eta \sigma &= \int d^2\bm{\xi}d^2\bm{\gamma} \chi_\rho (\bm{\xi},s)\chi_\sigma(\bm{\gamma},s)E(-\bm{{\xi'}},-s) \underbrace{\tr{E(-\bm{\gamma'},-s)}}_{\delta(\bm{\gamma'})},\\
        &= \int d^2\bm{\xi}\chi_\rho(\bm{\xi})\chi_\sigma(\sqrt{1-\eta}\bm{\xi}/\sqrt{\eta},s)E(-\bm{\xi}/\sqrt{\eta},-s),
    \end{align}
    with the following characteristic equation,
    \begin{align}
        \chi_{\rho \boxplus_\eta \sigma}(\bm{\gamma},s)=\int d^2\bm{\xi}\chi_\rho(\bm{\xi})\chi_\sigma(\sqrt{1-\eta}\bm{\xi}/\sqrt{\eta},s)\underbrace{\tr{D(-\bm{\gamma},s)E(-\bm{\xi}/\sqrt{\eta},-s)}}_{\delta(\bm{\gamma}-\bm{\xi}/\sqrt{\eta})}= \chi_\rho(\sqrt{\eta}\bm{\xi},s)\chi_\sigma(\sqrt{1-\eta}\bm{\xi},s),
    \end{align}
    giving us our desired result.
\end{proof}
\subsection{Fermionic Central Limit Theorem (CLT)}
For the remainder of these proofs we will assume that $s=0$ for simplicity. To find the convergence rate of the fermionic CLT in trace distance we will need the following lemma:
\begin{lemma}[Cumulant Bound]\label{app:lemma_cumulant_bound}
    For a given nonempty subset $K\subseteq[2m]$ with cardinality $|K|=k$ we have the following bound on the corresponding cumulant
    \begin{align}
        |C_K| \le k^{2k}.
    \end{align}
\end{lemma}
\begin{proof}
    The goal is to bound the size of the cumulants in terms of the moments of $\rho$. To do so we can Taylor expand $\log[\chi_\rho(\bm{\xi})]$ and use Eq.~\eqref{app:eqn_moment_generating_function} to find
    \begin{align}
        \log[\chi_\rho(\bm{\xi})]=\log\left[1+\sum_{J \subseteq K}M_J \bm{\xi}^J\right] = \sum_{l=1}^k\frac{(-1)^{l-1}}{l}\left(\sum_{J\subseteq K}M_J \bm{\xi}^J\right)^l,
    \end{align}
    where we only keep track of $\xi^K$ terms that could contribute to $C_K$. When expanding the $l$-th powers of the sum only disjoint products of Grassmann numbers will result in an $\bm{\xi}^K$ after appropriate permutations. As such, let $\lambda \vdash_l K$ denote the ordered partitions $\lambda$ of $K$ with $l$ pieces and $\lambda \vdash K$ denote the set of all partitions of $K$. By summing over each $l$-th order partition we can isolate the $\bm{\xi}^K$ terms,
    \begin{align}
    C_K=\sum_{l=1}^k \frac{(-1)^{l-1}}{l}\sum_{\lambda \vdash_l K}\mathrm{Sgn}(\pi_\lambda)\prod_{i=1}^l M_{\lambda_i}
    \end{align}
    where $\mathrm{Sgn}(\pi_\lambda)$ is the sign of the permutation reordering $\lambda$ into $K$. Using the fact that $|M_J|\le 1$ for all $J\subseteq [2m]$ and triangle inequality,
    \begin{align}
    |C_K| \le \sum_{l=1}^k \frac{1}{l} \sum_{\lambda \vdash_l K} 1 = \sum_{l=1}^k \frac{l!}{l} |\lambda \vdash_l K|
    \le \sum_{l=1}^k \frac{k!}{k} |\lambda \vdash_l K| = \frac{k!}{k}|\lambda \vdash K|,
    \end{align}
    where the appearance of $l!$ is to count all permutations of $K$ into $l$ partitions. The number of partitions of $K$, $|\lambda \vdash K|$, is given by the Bell number, $B_k$, which is upper-bounded by~\cite{berend_improved_2010}
    \begin{align}
        B_k < \left(\frac{0.792 k}{\ln(k+1)}\right)^k,
    \end{align}
    giving us that $|C_K| \le k^{2k}$. It is worth noting the points where this bound is loose and could be tightened in future work. Recall that for even states $\rho$, $M_J=0$ for odd order $J$. This implies that for odd order $K$ that $C_K=0$ and that for even order $K$ the bound overestimates the number of non-zero terms. However, for our purposes, we are interested in showing that the convergence in trace distance scales as $\tfrac{1}{n}$ and is at most polynomial in the number of modes.
\end{proof}

\begin{theorem}[Convergence Rate for Fermionic CLT in Trace Distance]\label{app:thm_fermionic_CLT}
    Let $\rho$ be an $m$-mode pure fermionic state and $\rho_G$ be the Gaussian state with the same second-moment quantities as $\rho$. Then $\rho^{\boxplus n}\rightarrow \rho_G$ in trace distance ($\norm{\cdot}_1$) with a rate
    \begin{align}
        \norm{\rho^{\boxplus n}-\rho_G}_1 \le \mathcal{O}\left(\frac{m^{12}}{n}\right).
    \end{align}
\end{theorem}

\begin{proof}[Proof of Theorem~\ref{app:thm_fermionic_CLT}]\label{app:proof_fermionic_CLT}
We begin by representing $\rho^{\boxplus n}$ and $\rho_G$ in terms of their characteristic functions
    \begin{align}
        \norm{\rho^{\boxplus n}-\rho_G}_1
         &= \left\lVert \int d^2\bm{\xi} \left[\chi_{\rho^{\boxplus n}}(\bm{\xi})-\chi_{\rho_G}(\bm{\xi})\right]E(-\bm{\xi})\right\rVert_1
        =\norm{\int d^2\bm{\xi} \left[\chi_{\rho}(\bm{\xi}/\sqrt{n})^n-\chi_{\rho_G}(\bm{\xi})\right]E(-\bm{\xi})}_1,\\
        &=\norm{\int d^2 \bm{\xi} \left(e^{\sum_{4\leq |J|}n^{1-|J|/2}C_J \bm{\xi}^J}-1\right)e^{\sum_{|J|=2}C_J\bm{\xi}^J} E(-\bm{\xi})}_1,\label{app:eqn_trace_distance}
    \end{align}
    where the last equality comes from the Taylor expansion of the cumulant generating function and ignoring odd order cumulants which are zero,
    \begin{align}
        \log[\chi_{\rho^{\boxplus n}}(\bm{\xi})]=n\log[\chi_{\rho}(\bm{\xi}/\sqrt{n})]
        =n\sum_{J}\frac{1}{\sqrt{n}^{|J|}}C_J \bm{\xi}^J = \sum_{J}n^{1-|J|/2}C_J\bm{\xi}^J.
    \end{align}
    Due to the unitary invariance of the trace distance, we can assume that (without loss of generality) $\rho_G$ is diagonal. Therefore we have that $\exp\left[\sum_{|J|=2}C_J \bm{\xi}^J\right] = \exp\left[\sum_j -\beta_j \xi_j^*\xi_j\right]$ where $\beta_j$ is real-valued (the symplectic eigenvalues of the covariance matrix). To proceed we will make use of the following facts which follow immediately from Eq.~\eqref{app:eqn_ladder_E_operator}:
    \begin{align}\label{app:trace_norm_bounds}
    \begin{split}
        \norm{\int d^2 \bm{\xi} \exp\left[\sum_j \beta_j \xi_j^*\xi_j\right] E(-\bm{\xi})}_1\leq \norm{\rho_G}_1 \leq 1 &\quad \norm{\int d^2 \xi_j \xi_j\exp\left[ \beta_j \xi_j^*\xi_j\right] E(-\xi_j)}_1\leq \norm{a_j}_1 \leq 1,\\
        \norm{\int d^2 \xi_j \xi_j^*\exp\left[ \beta_j \xi_j^*\xi_j\right] E(-\xi_j)}_1\leq \norm{a_j}_1 \leq 1 &\quad \norm{\int d^2 \xi_j \xi_j \xi_j^*\exp\left[ \beta_j \xi_j^*\xi_j\right] E(-\xi_j)}_1 \leq \norm{-2\{a_j^\dagger a_j\}}_1 \leq 1.
    \end{split}
    \end{align}
    Consider a single value of $e^{C_J \bm{\xi}^J}$ contained within the exponential in Eq.~\eqref{app:eqn_trace_distance}, we can bound its contribution to the trace norm as follows
    \begin{align}
        &\norm{\int d^2\bm{\xi} \left(e^{C_J \bm{\xi}^J}-1\right)e^{\sum_j -\beta_j \xi_j^*\xi_j} E(-\bm{\xi})}_1 = \norm{\int d^2 \bm{\xi} C_J \bm{\xi}^J e^{\sum_j -\beta_j \xi_j^*\xi_j} E(-\bm{\xi})}_1,\\
        &= \norm{\int d^2 \bm{\xi}^J C_J \bm{\xi}^J \exp\left[\sum_{j\in J} -\beta_j \xi_j^*\xi_j\right] E(-\bm{\xi}^J)}_1 \norm{\int d^2 \bm{\xi}^{J^C} \exp\left[\sum_{j\in {J^C}} -\beta_j \xi_j^*\xi_j\right] E(-\bm{\xi}^{J^C})}_1,\\
        &\le |C_J| \norm{\int d^2 \bm{\xi}^J \bm{\xi}^J \exp\left[\sum_{j\in J} -\beta_j \xi_j^*\xi_j\right] E(-\bm{\xi}^J)}_1 \cdot 1 \le |C_J| \le e^{|C_J|}-1.
    \end{align}
    Let's walk through the intuition while defining the notation utilized above to prove this result. We can Taylor expand $e^{C_J \bm{\xi}^J}$ to obtain the first equality then use the fact that $\norm{A \ot B}_1 = \norm{A}_1 \norm{B}_1$ to break apart the integral. The act of breaking apart the integral is done mode-wise, so we need to group by modes affected by $J$ and those that are not. We denote this set by $\bar{J}\subseteq [m]$ where for $J=\{1,2,4\}\subseteq [6]$ this gives $\bar{J}=\{1,2\}\subseteq [3]$. This motivates the notation $d^2 \bm{\xi}^{J}=\prod_{j\in \bar{J}} d\xi_j^* d\xi_j$ and $E(-\bm{\xi}^J)=\prod_{j\in \bar{J}}E(-\xi_j)$. As a slight abuse of notation, we define $\bar{J}^C=[m]\setminus \bar{J}$ and utilize the notation $d^2\bm{\xi}^{J^C}=\prod_{j\in \bar{J}^C} d\xi_j^* d\xi_j$ and $E(-\bm{\xi}^{J^C})=\prod_{j\in \bar{J}^C}E(-\xi_j)$. With this, we can properly split up the integral and recognize that the $J^C$ component corresponds to a Gaussian state and is thus bounded above by $1$. Decomposing the $J$ component further and applying the bounds in Eq.~\eqref{app:trace_norm_bounds} gives that the whole quantity is bounded above by $|C_J|$. However, we loosen the bound further to $e^{|C_J|}-1$ as will become apparent as we move into bounding two terms. Now we consider bounding a term with $e^{C_J\bm{\xi}^J + C_K \bm{\xi}^K}$ rather than $e^{C_J \bm{\xi}}$ where we can see through analogous arguments that
    \begin{align}
        \norm{\int d^2\bm{\xi} \left(e^{C_J \bm{\xi}^J+C_K\bm{\xi}^K}-1\right)e^{\sum_j -\beta_j \xi_j^*\xi_j} E(-\bm{\xi})}_1 &\leq e^{|C_J|+|C_K|}-1.
    \end{align}
    The key idea is that $e^{C_J \bm{\xi}^J + C_K \bm{\xi}^K}=e^{C_J \bm{\xi}^J}e^{C_K \bm{\xi}^K}$ since only even order $C_J$ are nonzero and must commute with each other. By induction, this implies that
    \begin{align}
        \norm{\rho^{\boxplus n}-\rho_G}_1 \leq \exp\left[\sum_{4\leq|J|}n^{1-|J|/2}|C_J|\right]-1.
    \end{align}
    Trivially we have that $\norm{\rho^{\boxplus n}-\rho_G}_1\le 2$ by triangle inequality, thus for the bound to be nontrivial, we need $0\le x\le \log (3)$ in $e^{x}-1$. In this range, we have that $e^{x}-1 \le \tfrac{2}{\log (3)}x$, which implies that a bound on $x$ will give a bound on the trace distance up to constants. Focusing on bounding this term we can see that
    \begin{align}
        \sum_{4\leq|J|}n^{1-|J|/2}|C_J|&\le \sum_{k=4}^{2m}\sum_{|J|=k}n^{1-k/2}k^{2k} \leq \sum_{k=4}^{2m}(2m)^k n^{1-k/2}k^{2k} \leq \sum_{k=4}^{2m}(2m)^k n^{1-k/2}{2m}^{2k} = \sum_{k=4}^{2m} \left(\frac{(2m)^3}{\sqrt{n}}\right)^{k} n.
    \end{align}
    Observe that for fixed $k>2$ and $k$-dependent constant $0\le q_k$, there always exists some $N_k\in \mathbb{N}$ such that for all $n\ge N_k$ we have that $q_k/n^k \le 1/n$. Therefore only the first term in the sum matters asymptotically leading us to the desired bound of $\mathcal{O}(m^{12}/n)$.
\end{proof}
With this, we have demonstrated that the state with the same second-order cumulants as $\rho$ and $\rho^{\boxplus n}$ coincide in the limit $n\to \infty$. Finally, we must show that the closest Gaussian in relative entropy to $\rho$ coincides with these other two notations of the closest Gaussian.
\subsection{Closest Gaussian in Relative Entropy}
\begin{theorem}\label{app:relative_entropy_Gaussian}
    Let $\rho_G$ be the Gaussian state with the same second-order cumulants as $\rho$. Then $\rho_G = \min_{\rho_G'}S(\rho || \rho_G')$ where $S(\cdot||\cdot)$ denotes the relative entropy.
\end{theorem}
\begin{proof}
    Let the closest Gaussian state to $\rho$ in relative entropy be denoted $\rho_G'$ and $\rho_G$ is the Gaussian with the same second-order cumulants as $\rho$. Since $\rho_G'$ is Gaussian, it is the thermal state of a Hamiltonian $H$ quadratic in raising and lowering operators and expressible as $\rho_G'=\exp(H')/Z(H')$ where $Z(H')=\tr{\exp(H')}$ is the partition function or normalization factor. Further, for any quadratic operator $Q$, we have that $\tr{\rho Q}=\tr{\rho_G Q}$ as all second-order quantities are captured by $\rho_G$. Using the definition of relative entropy,
    \begin{align}
    \begin{split}
        S(\rho || \rho_{G'}) &= \tr{\rho \log \rho}-\tr{\rho \log e^{H'}}+\log Z(H')= - \tr{\rho H'}+\log Z(H') - S(\rho),\\
        &= -\tr{\rho_G H'}+\log Z(H')-S(\rho)=-\tr{\rho_G \log \rho_G'}-S(\rho),\\
        &= -\tr{\rho_G \log \rho_G'}+S(\rho_G)-S(\rho_G)-S(\rho),\\
        &= S(\rho_G)-S(\rho)+S(\rho_G||\rho_G'),
    \end{split}
    \end{align}
    then minimizing over Gaussian distributions we have that
    \begin{align}
        \min_{H'} S(\rho || \rho_G')=&S(\rho_G)-S(\rho)
        +\min_{H'} S(\rho_G||\rho_G')=S(\rho_G)-S(\rho),
    \end{align}
    which is achieved by letting $\rho_G = \rho_G'$ as $S(\rho_G ||\rho_G')\ge 0$ and zero if and only if $\rho_G=\rho_{G'}$.
\end{proof}
Thus, all three Gaussians coincide and we can refer to $\rho_G$ without ambiguity.
\section{Measures of Non-Gaussian Magic}
After showing that all three Gaussian coincide, we can now develop multiple measures of non-Gaussian magic by measuring the difference between them and $\rho$. First, notice that in the proof of Theorem~\ref{app:relative_entropy_Gaussian} we showed that $S(\rho||\rho_G)=S(\rho_G)-S(\rho)$, thus for pure $\rho$ measuring the relative entropy to $\rho_G$ reduces to measuring $S(\rho_G)$. Therefore, $S(\rho||\rho_G)=S(\rho_G)$ is a measure of non-Gaussian magic as $S(\rho_G)=0$ if and only if $\rho$ is a pure Gaussian. This quantity can be calculated by measuring all second-order cumulants of $\rho$. We will now derive other measures of non-Gaussian magic through the violation of the matchgate identity, the violation of Wick's theorem, and the SWAP test.
\subsection{Matchgate Identity}
The statement that Gaussians are uniquely invariant under self-convolution means that Gaussians are unchanged by the fermionic beam splitter. That is, for all time $t$, $U(t)\ket{\psi}\ket{\psi}=\ket{\psi}\ket{\psi}$ where $U(t)=\exp(it\Lambda)$ and $\Lambda=i\Gamma=i\sum_{i}a_i^\dagger b_i +a_ib_i^\dagger$ is the generator of the fermionic beam splitter. Taylor expanding $U(t)=\mathbb{I}+it\Lambda+\mathcal{O}(t^2)$ implies for small $t$ and a pure Gaussian state $\ket{\psi}$
\begin{align}
    U(t)\ket{\psi}^{\ot 2} &\approx \ket{\psi}^{\ot 2}+it\Lambda \ket{\psi}^{\ot 2} = \ket{\psi}^{\ot 2},\\
    \implies \Lambda \ket{\psi}^{\ot 2}&=0.
\end{align}
Conversely, suppose $\Lambda\ket{\psi}^{\ot2}=0$, then $U(t)\ket{\psi}^{\ot 2}=\ket{\psi}^{\ot 2}$ for all time $t$ and $\ket{\psi}$ is Gaussian. Therefore we have arrived at the matchgate identity: $\Lambda \ket{\psi}\ket{\psi}=0$ if and only if $\ket{\psi}$ is Gaussian. Notice that here the matchgate identity is derived from convolution alone. Naturally, by computing the deviation of $\Lambda \ket{\psi}\ket{\psi}$ from zero we see that $\norm{\Lambda \ket{\psi}\ket{\psi}}_2$ is a new measure of non-Gaussian magic. However, since $\Lambda$ is quadratic, it will only depend on second-order cumulants. This fact is captured in the following lemma.
\begin{theorem}[Violation of Matchgate Identity]\label{app:matchgate}
    We have the following equality
    \begin{align}
        \norm{\Lambda \ket{\psi}\ket{\psi}}_2^2 = \frac{m}{2}-2\sum_{|J|=2}|M_J|^2 = \tfrac{m}{2}-\norm{\Sigma}_2^2
    \end{align}
    where $(\Sigma)_{ij}=\tr{\{\bm{a}_i\bm{a}_j\}\rho}$ is the covariance matrix for $\rho$.
\end{theorem}
\begin{proof}
    Let $\Gamma_i = a_i^\dagger b_i +a_i b_i^\dagger$, then
    \begin{align}
        \norm{\Lambda \ket{\psi}\ket{\psi}}_2^2 = \tr{\Lambda^2 \rho^{\ot 2}}=-\sum_{i,j}\tr{\Gamma_i\Gamma_j \rho^{\ot 2}}.
    \end{align}
    Expanding out $\Gamma_i\Gamma_j$ we see that
    \begin{align}
        \Gamma_i\Gamma_j &= a_i^\dagger b_i a_j^\dagger b_j + a_i^\dagger b_i a_jb_j^\dagger + a_ib_i^\dagger a_j^\dagger b_j + a_ib_i^\dagger a_jb_j^\dagger,\\
        &= -a_i^\dagger a_j^\dagger b_i b_j -a_i^\dagger a_j b_i b_j^\dagger - a_ia_j^\dagger b_i^\dagger b_j - a_ia_jb_i^\dagger b_j^\dagger,\\
        &= -(a_i^\dagger a_j^\dagger)(b_j^\dagger b_i^\dagger )^\dagger-(a_i^\dagger a_j)(b_jb_i^\dagger)^\dagger-(a_ia_j^\dagger)(b_j^\dagger b_i)^\dagger-(a_ia_j)(b_jb_i)^\dagger.
    \end{align}
    Assuming for the moment that $i\neq j$, we have the following,
    \begin{align}
        \tr{\Gamma_i\Gamma_j \rho^{\ot 2}}&=-\tr{a_i^\dagger a_j^\dagger \rho}\tr{(b_j^\dagger b_i^\dagger )^\dagger \rho}-\tr{a_i^\dagger a_j \rho}\tr{(b_jb_i^\dagger )^\dagger \rho}\\
        &-\tr{a_ia_j^\dagger \rho}\tr{(b_j^\dagger b_i)^\dagger \rho}-\tr{a_ia_j \rho}\tr{(b_jb_i)^\dagger \rho},\\
        &= -\tr{a_i^\dagger a_j^\dagger \rho}\overline{\tr{(b_j^\dagger b_i^\dagger )\rho}} +\cdots\\
        &=\left|\tr{a_i^\dagger a_j^\dagger \rho}\right|^2+\left|\tr{a_i^\dagger a_j \rho}\right|^2 + \left|\tr{a_ia_j^\dagger \rho}\right|^2+\left|\tr{a_ia_j\rho}\right|^2
    \end{align}
    where we used the fact that $i\neq j$ to anti-commute the $b$ terms into the same ordering as the $a$ terms. These terms correspond to the squared magnitude of second-order moments. Two of each type of term will appear for the cases of $i<j$ and $i>j$ where we are free to switch the order of $i$ and $j$ in either case due to us having the square of the magnitude. Now suppose $i=j$, then
    \begin{align}
        \Gamma_i^2 &= -(a_i^\dagger a_i)(b_ib_i^\dagger) -(a_ia_i^\dagger)(b_i^\dagger b_i)=-(a_i^\dagger a_i)(1-b_i^\dagger b_i)-(1-a_i^\dagger a_i)(b_i^\dagger b_i)=-a_i^\dagger a_i - b_i^\dagger b_i +2a_i^\dagger a_ib_i^\dagger b_i,\\
        &=-\{a^\dagger_ia_i\}-\tfrac{1}{2}-\{b_i^\dagger b_i\}-\tfrac{1}{2}+2(\{a_i^\dagger a_i\}+\tfrac{1}{2})(\{b_i^\dagger b_i\}+\tfrac{1}{2}),\\
        &=2\{a^\dagger_i a_i\}\{b^\dagger_ib_i\}-\tfrac{1}{2},
    \end{align}
    where we recall that $\{a_i^\dagger a_i\}=a_i^\dagger a_i-\tfrac{1}{2}=-\{a_ia_i^\dagger\}$. Therefore,
    \begin{align}
        \tr{\Gamma_i^2 \rho^{\ot 2}} = 2\tr{\{a_i^\dagger a_i\}\rho}\tr{\{b_i^\dagger b_i\}\rho}-\tfrac{1}{2} = 2\left|\tr{\{a_i^\dagger a_i\}\rho }\right|^2 -\tfrac{1}{2}.
    \end{align}
    Combining all of these expressions together we obtain the following:
    \begin{align}
         \tr{\Lambda^2 \rho^{\ot 2}} &= -\sum_{i,j}\tr{\Gamma_i\Gamma_j \rho^{\ot 2}} = -\sum_{i\neq j}\tr{\Gamma_i\Gamma_j\rho^{\ot 2}}-\sum_{i=j}\tr{\Gamma_i^2 \rho^{\ot 2}},\\
        &=-\sum_{i\neq j}\left(\left|\tr{a_i^\dagger a_j^\dagger \rho}\right|^2+\left|\tr{a_i^\dagger a_j \rho}\right|^2 + \left|\tr{a_ia_j^\dagger \rho}\right|^2+\left|\tr{a_ia_j\rho}\right|^2 \right)-\sum_{i=j}\left(2\left|\tr{\{a_i^\dagger a_i\}\rho }\right|^2 -\tfrac{1}{2}\right),\\
        &=\frac{m}{2}-2\sum_{|J|=2}|M_J|^2.
    \end{align}
    Recall that by convention the sum over $J$ is over the subsets of $[2m]$ with cardinality 2 where elements are expressed in increasing order such that only $\{1,2\}$ will appear but $\{2,1\}$ will not also appear to avoid double counting. This accounts for the factor of $2$ outside the sum. If we define $(\Sigma)_{ij}=\tr{\{\bm{a}_i\bm{a}_j\}\rho}$ to be the $2m\times 2m$ covariance matrix for $\rho$ then we have that
    \begin{align}
        \norm{\Lambda \ket{\psi}\ket{\psi}}_2^2 = \tr{\Lambda^2 \rho^{\ot 2}}=\frac{m}{2}-2\sum_{|J|=2}|M_J|^2 = \tfrac{m}{2}-\norm{\Sigma}_2^2
    \end{align}
    where $\norm{\cdot}_2$ denotes the Schatten $2$-norm leading to the desired result.
\end{proof}
It is now clear that $\norm{\Lambda \ket{\psi}\ket{\psi}}_2^2$ only depends on the second-order cumulants of $\rho$ or truly $\rho_G$. Therefore, $\norm{\Lambda\ket{\psi}\ket{\psi}}_2^2$ and $S(\rho_G)$ capture the same information about non-Gaussian magic. However, rather than estimating all second-order cumulants to compute $S(\rho_G)$, we can use the following lemma instead.
\begin{lemma}\label{appendix:lemma_second_order_matchgate}
To compute the violation of the matchgate identity, we can estimate the following quantity
    \begin{align}
        \frac{d^2}{dt^2}\norm{(\exp(it\Lambda)-\mathbb{I})\ket{\psi}\ket{\psi}}_2 \Big|_{t=0}=\norm{\Lambda \ket{\psi}\ket{\psi}}_2^2.
    \end{align}
\end{lemma}
\begin{proof}
    Observe that for small $t$ we have that,
    \begin{align}
        \exp(it\Lambda)-\mathbb{I}=it\Lambda-\frac{t^2}{2}\Lambda^2+\OC(t^3).
    \end{align}
    Thus,
    \begin{align}
        \norm{(\exp(it\Lambda)-\mathbb{I})\ket{\psi}\ket{\psi}}_2^2=\bra{\psi}\bra{\psi}(\exp(-it\Lambda)-\mathbb{I})(\exp(it\Lambda)-\mathbb{I})\ket{\psi}\ket{\psi}=t^2\bra{\psi}\bra{\psi}\Lambda^2\ket{\psi}\ket{\psi}+\OC(t^3).
    \end{align}
    Therefore after taking the second derivative and evaluating it at zero, we arrive at the desired result.
\end{proof}

\subsection{Wick's Theorem}\label{app:Wicks_Cumulants}
Since any Gaussian state is completely characterized by its covariance matrix $(\Sigma)_{ij}=\tr{\{\bm{a}_i\bm{a}_j\rho\}}$ any high-order moment can be decomposed into sums of products of lower-order moments. This is precisely the statement of Wick's theorem~\cite{wick_evaluation_1950} in QFT or Isserlis' theorem~\cite{isserlis_formula_1918} in probability theory. Wick's theorem is the result of expanding the following expression for a pure Gaussian state $\rho$:
\begin{align}
    M_J &= \frac{\partial }{\partial \bm{\xi}^J}\chi_\rho(\bm{\xi})\bigg\vert_{\bm{\xi}=0}=\frac{\partial }{\partial \bm{\xi}^J}\exp\left[\log\chi_\rho(\bm{\xi})\right]\bigg\vert_{\bm{\xi}=0}=\frac{\partial }{\partial \bm{\xi}^J}\exp\left[\sum_{|K|=2}C_K\bm{\xi}^K\right]\Bigg\vert_{\bm{\xi}=0}.
\end{align}
Concretely, for a nonempty subset $J\subseteq [2m]$ the expansion for $M_J$ of a Gaussian state $\rho$ is given by the Pfaffian of the submatrix $\Sigma\vert_J$ of the covariance matrix (i.e. the rows and columns with indices given by $J$). In this work, the Pfaffian of a $2m\times 2m$ anti-symmetric matrix matrix is defined as
\begin{align}
    \mathrm{Pf}[\Sigma]=\frac{1}{ 2^n n!}\sum_{\sigma \in S_{2n}}\mathrm{Sng}(\sigma)\Sigma_{\sigma_1,\sigma_2}\cdots \Sigma_{\sigma_{2n-1},\sigma_{2n}}.
\end{align}
For instance, with $m=2$ modes, the covariance matrix is
\begin{align}
    \Sigma  = \begin{bmatrix}
        0 & M_{12} & M_{13} & M_{14}\\
        -M_{12} & 0 & M_{23} & M_{24} \\
        -M_{13} & -M_{23} & 0 & M_{34}\\
        -M_{14} & -M_{24} & -M_{34} &0
    \end{bmatrix}.
\end{align}
Using Wick's theorem we have that
\begin{align}
    M_{\{1,2,3,4\}} = \mathrm{Pf} \begin{bmatrix}
        0 & M_{12} & M_{13} & M_{14}\\
        -M_{12} & 0 & M_{23} & M_{24} \\
        -M_{13} & -M_{23} & 0 & M_{34}\\
        -M_{14} & -M_{24} & -M_{34} &0
    \end{bmatrix} = M_{12}M_{34}-M_{13}M_{24}+M_{14}M_{23},
\end{align}
and
\begin{align}
    M_{\{1,2,3\}} = \mathrm{Pf} \begin{bmatrix}
        0 & M_{12} & M_{13} \\
        -M_{12} & 0 & M_{23}  \\
        -M_{13} & -M_{23} & 0 \\
    \end{bmatrix} = 0,
\end{align}
and
\begin{align}
    M_{\{2,3\}} = \mathrm{Pf} \begin{bmatrix}
 0 & M_{23}  \\
-M_{23} & 0 \\
    \end{bmatrix} = M_{23}.
\end{align}
A more convenient way to do these computations is diagrammatically using Feynman diagrams. For instance, with $m=3$ we can arrange the numbers $1-6$ at the corners of a hexagon and count the number of times two lines cross to determine the sign. If the number of crossings is even then the sign is positive, if it's odd then the sign is negative. The number of terms corresponds to the rotational symmetries of each diagram. We are in essence treating each node as a Grassmann number initially in order $123456\mapsto\xi_1^*\xi_1\xi_2^*\xi_2\xi_3^*\xi_3$ which we permute into each possible ordering.
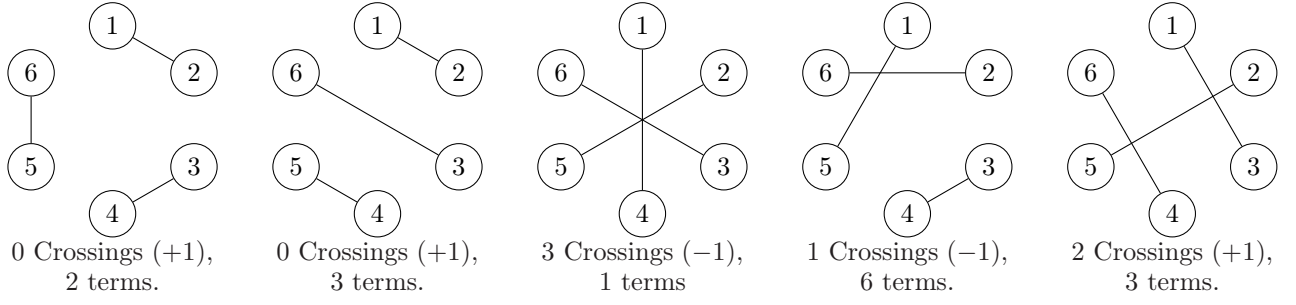
\begin{center}
    % First Diagram
    \begin{minipage}{0.19\textwidth}
        \centering
        \begin{tikzpicture}
            % Reduced radius for the hexagon
            \node[circle,draw] (1) at (90:1.25) {1};
            \node[circle,draw] (2) at (30:1.25) {2};
            \node[circle,draw] (3) at (330:1.25) {3};
            \node[circle,draw] (4) at (270:1.25) {4};
            \node[circle,draw] (5) at (210:1.25) {5};
            \node[circle,draw] (6) at (150:1.25) {6};

            \draw (1) -- (2);
            \draw (3) -- (4);
            \draw (5) -- (6);
        \end{tikzpicture}
        \vspace{0.3cm}
        \parbox{\textwidth}{\centering $0$ Crossings $(+1)$,\\ $2$ terms.}
    \end{minipage}
    % Second Diagram
    \begin{minipage}{0.19\textwidth}
        \centering
        \begin{tikzpicture}
            \node[circle,draw] (1) at (90:1.25) {1};
            \node[circle,draw] (2) at (30:1.25) {2};
            \node[circle,draw] (3) at (330:1.25) {3};
            \node[circle,draw] (4) at (270:1.25) {4};
            \node[circle,draw] (5) at (210:1.25) {5};
            \node[circle,draw] (6) at (150:1.25) {6};

            \draw (1) -- (2);
            \draw (6) -- (3);
            \draw (5) -- (4);
        \end{tikzpicture}
        \vspace{0.3cm}
        \parbox{\textwidth}{\centering $0$ Crossings $(+1)$,\\ $3$ terms.}
    \end{minipage}
    % Third Diagram
    \begin{minipage}{0.19\textwidth}
        \centering
        \begin{tikzpicture}
            \node[circle,draw] (1) at (90:1.25) {1};
            \node[circle,draw] (2) at (30:1.25) {2};
            \node[circle,draw] (3) at (330:1.25) {3};
            \node[circle,draw] (4) at (270:1.25) {4};
            \node[circle,draw] (5) at (210:1.25) {5};
            \node[circle,draw] (6) at (150:1.25) {6};

            \draw (1) -- (4);
            \draw (5) -- (2);
            \draw (6) -- (3);
        \end{tikzpicture}
        \vspace{0.3cm}
        \parbox{\textwidth}{\centering $3$ Crossings $(-1)$,\\ $1$ terms}
    \end{minipage}
    % Fourth Diagram
    \begin{minipage}{0.19\textwidth}
        \centering
        \begin{tikzpicture}
            \node[circle,draw] (1) at (90:1.25) {1};
            \node[circle,draw] (2) at (30:1.25) {2};
            \node[circle,draw] (3) at (330:1.25) {3};
            \node[circle,draw] (4) at (270:1.25) {4};
            \node[circle,draw] (5) at (210:1.25) {5};
            \node[circle,draw] (6) at (150:1.25) {6};

            \draw (1) -- (5);
            \draw (2) -- (6);
            \draw (3) -- (4);
        \end{tikzpicture}
        \vspace{0.3cm}
        \parbox{\textwidth}{\centering $1$ Crossings $(-1)$,\\ $6$ terms.}
    \end{minipage}
    % Fifth Diagram
    \begin{minipage}{0.19\textwidth}
        \centering
        \begin{tikzpicture}
            \node[circle,draw] (1) at (90:1.25) {1};
            \node[circle,draw] (2) at (30:1.25) {2};
            \node[circle,draw] (3) at (330:1.25) {3};
            \node[circle,draw] (4) at (270:1.25) {4};
            \node[circle,draw] (5) at (210:1.25) {5};
            \node[circle,draw] (6) at (150:1.25) {6};

            \draw (1) -- (3);
            \draw (2) -- (5);
            \draw (6) -- (4);
        \end{tikzpicture}
        \vspace{0.3cm}
        \parbox{\textwidth}{\centering $2$ Crossings $(+1)$,\\ $3$ terms.}
    \end{minipage}
\end{center}
One can use a similar mathematical derivation to find an expression for the cumulants in terms of the moments of the state:
\begin{align}
    C_K = \frac{\partial }{\partial \bm{\xi}^K}\log [\chi_\rho(\bm{\xi})]\bigg\vert_{\bm{\xi}=0}=\frac{\partial }{\partial \bm{\xi}^K}\log\left[1+\sum_{1\le |J|} M_J \bm{\xi}^J\right]\Bigg\vert_{\bm{\xi}=0}=\sum_{l=1}^k \frac{(-1)^{l-1}}{l}\sum_{\lambda \vdash_l K}\mathrm{Sgn}(\pi_\lambda)\prod_{i=1}^l M_{\lambda_i}
\end{align}
where $\lambda \vdash_l K$ denotes the ordered partitions $\lambda$ of $K$ with $l$ pieces and $\mathrm{Sgn}(\pi_\lambda)$ is the sign of the permutation reordering $\lambda$ into $K$. This formula arises from brute force and inspection however, it resembles the one from classical probability theory~\cite{peccati_wiener_2011} with an additional $\mathrm{Sgn}(\pi_\lambda)$ to account for anti-commutation relations. To measure non-Gaussian magic, one can directly estimate $C_J$, however, the number of terms that must be computed grows rapidly. Thus, low-order cumulants or weak interactions are what can be measured efficiently in practice.

\subsection{SWAP Test}\label{app:sec_SWAP_test}
So far, relative entropy and the matchgate identity only depend on second-order cumulants, in contrast, the violation of Wick's theorem captures magic in higher-order cumulants but is challenging to compute. Utilizing the fact that $\rho^{\boxplus 2}=\rho$ iff $\rho$ is a pure Gaussian state, we can measure the difference between $\rho^{\boxplus 2}$ and $\rho$ via the SWAP test. To start, let's focus on the action of SWAP on two single modes. By inspecting the behavior of SWAP on the computational basis, we see that it can be expressed in terms of ladder operators as
\begin{align}
    \mathrm{SWAP} = a^\dagger b +ab^\dagger +a^\dagger a b^\dagger b+aa^\dagger bb^\dagger= a^\dagger b +ab^\dagger+2\{a^\dagger a\}\{b^\dagger b\}+\tfrac{1}{2}\mathbb{I},
\end{align}
where $\{a^\dagger a\}=a^\dagger a -\tfrac{1}{2}$. Therefore, functionally the SWAP test has the following form
\begin{align}
    \tr{\mathrm{SWAP}\rho \ot \rho} &= \tr{\prod_{j=1}^m [a^\dagger_i b_i+a_ib_i^\dagger+2\{a_i^\dagger a_i\}\{b_i^\dagger b_i\}+\tfrac{1}{2}\mathbb{I}_i]\rho \ot \rho}=\sum_{J\subseteq [2m]}\delta(J)|M_J|^2,
\end{align}
where $\delta(J)$ is some function of $J$. To see why it has this form, observe that expanding out this product results in any order moment to appear as we can choose to either apply $\mathbb{I}_i,a_i,a_i^\dagger,\{a_i^\dagger a_i\}$ on system $A$, while system $B$ will always receive the adjoint of that choice. Taking the adjoint for system $B$ will reverse the ordering, however, since every term will at worst anti-commute, this will only result in an overall sign on that term. Thus, $\delta(J)$ records these anti-commutations and the additional factors that appear depending on the moment. Crucially, $\delta(J)$ is independent of $\rho$ and purely a combinatorial object. When performing the SWAP test on $\rho$ and $\rho^{\boxplus 2}$ we can use the quantum convolution theorem (Theorem~\ref{app:thm_fermionic_convolution_theorem}) and the chain rule to see that
\begin{align}
    \chi_{\rho^{\boxplus 2}}(\bm{\xi})=\chi_\rho(\bm{\xi}/\sqrt{2})^2 \quad M_J'=\tfrac{\partial}{\partial \bm{\xi}^J}\chi_{\rho^{\boxplus 2}}(\bm{\xi})\Big\vert_{\bm{\xi}=0}=2^{1-|J|/2}M_J
\end{align}
which implies
\begin{align}
    \tr{\mathrm{SWAP}\rho \ot \rho^{\boxplus 2}}=\sum_{J\subseteq [2m]}\delta(J)2^{1-|J|/2}|M_J|^2.
\end{align}
Since $\rho$ is assumed to be pure, we have that $\tr{\mathrm{SWAP}\rho \ot \rho}=\tr{\rho^2}=1$ such that $|\tr{\mathrm{SWAP}\rho\ot\rho^{\boxplus 2}}-1|$ is a measure of non-Gaussian magic that captures information about higher than second-order cumulants. Indeed, we have that $\sum_{|J|>2}\delta(J)|M_J|^2=1-\sum_{|J|\le 2}\delta(J)|M_J|^2$ whenever $\rho$ is pure, thus moments up to second-order contain information about moments strictly greater than second-order. The SWAP test utilizes this fact to extract information about high-order moments and thus high-order cumulants while being efficient to estimate.

\section{Connection to Work By Lyu and Bu}\label{app:lyu_connection}
During the preparation of this work, the authors became aware of another paper by Lyu and Bu~\cite{lyu_fermionic_2024} studying fermionic convolution, the central limit theorem, and the application to measures of non-Gaussian magic. Here we outline the connection between the approaches taken by the two works through the formulation of the characteristic functions.

The begin, let $\{\gamma_j\}_{j=1}^{2m}$ denote the set of $2m$ Hermitian Majorana operators expressed in terms of ladder operators as
\begin{align}
    \gamma_{2j-1}=a_j+a_j^\dagger \quad \gamma_{2j}=i(a_j-a_j^\dagger),
\end{align}
with $\{\gamma_j,\gamma_k\}=2\delta_{jk}\mathbb{I}$. This is in effect a change of basis similar to that of using $(x,p)$ or $(a^\dagger,a)$ in bosonic systems with the caveat of an additional factor of $\tfrac{1}{\sqrt{2}}$,
\begin{align}
    \frac{1}{\sqrt{2}}\begin{bmatrix}
        \gamma_{2j-1} \\ \gamma_{2j}
    \end{bmatrix}=\frac{1}{\sqrt{2}}\begin{bmatrix}
        1 & 1 \\ i & -i
    \end{bmatrix}
    \begin{bmatrix}
        a_j^\dagger \\ a_j
    \end{bmatrix}=\begin{bmatrix}
        \tilde{\gamma}_{2j-1} \\ \tilde{\gamma}_{2j}
    \end{bmatrix}.
\end{align}
To avoid confusion from different notation and conventions used between the two works, let $\gamma$ and $\eta$ exclusively refer to the Majorana operators and Grassmann numbers utilized in~\cite{lyu_fermionic_2024} while $\tilde{\gamma}$ and $\tilde{\eta}$ will exclusively refer to the rescaled Majorana operators defined above and Grassmann numbers as defined in this work. The only difference between the definition of Grassmann numbers $\eta$ and $\tilde{\eta}$, is the enforcing of $\eta=\eta^*$ in~\cite{lyu_fermionic_2024} versus $\tilde{\eta}\neq \tilde{\eta}^*$ in this work. Just as bosonic operators $(x,p)$ quantize $\mathbb{R}^2$ with both being commutative, uncountable, and Hermitian or $(a^\dagger,a)$ quantizing $\mathbb{C}$ with both being commutative, uncountable, but not Hermitian, the same analogy plays out here. As seen throughout this work, the main utility of the characteristic function or moment-generating function is to Taylor expand and associate to each moment a number that recovers the original algebraic properties of the system. This is done by taking the Fourier transform of the density matrix $\rho$. Again, the distinction is in how the Fourier transform is taken. The classical Fourier transform of $f\in L^1(\mathbb{C}^m)$ is of the form
\begin{align}
\FC[f](\bm{t})=\int e^{i\langle \bm{t},\bm{x}\rangle}f(\bm{x})d\bm{x}
\end{align}
where $e^{i\langle \bm{t}, \bm{x}\rangle}$ is known as the kernel. In the generalized theory of harmonic analysis, we can take the Fourier transform of a function $f\in L^1(G)$ where $G$ is a topological group by replacing the transition invariant Lebesgue measure $d\bm{x}$ with the analogous Haar measure $d\mu$. The kernel is (loosely speaking) replaced by the "trace" of the representation $G$ in its Fourier dual $\hat{G}$. Since we are utilizing two different representations we must use the two different kernels,
\begin{align}
    D(\bm{\eta})=\exp\left( \sum_{j=1}^{2m} \gamma_j \eta_j \right) \quad \tilde{D}(\bm{\tilde{\eta}})=\exp\left(\sum_{j=1}^m a_j^\dagger \tilde{\eta}_j + a_j\tilde{\eta}_j^*\right),
\end{align}
where the former is utilized in~\cite{lyu_fermionic_2024} while the latter is the fermionic analog of the Weyl operator utilized in this work. Then the moment-generating function of $\rho$ is given by
\begin{align}
\chi_\rho(\bm{\eta})=\tr{D(\bm{\eta})\rho} \quad \tilde{\chi}_\rho(\tilde{\bm{\eta}})=\tr{\tilde{D}(\tilde{\bm{\eta}})\rho}.
\end{align}
\end{document}